\documentclass[article]{article}
\usepackage{qic,epsfig}
\usepackage{amsmath}
\usepackage{amsfonts}
\usepackage{enumerate}
\usepackage{amssymb}
\usepackage{xcolor}

\usepackage{ifdraft}

\newcommand{\remove}[1]{}

\newcommand{\p}{\mathsf{P}}
\newcommand{\binset}{\{0,1\}}
\newcommand{\BQP}{\mathsf{BQP}}
\newcommand{\PreciseBQP}{\mathsf{PreciseBQP}}
\newcommand{\PostBQP}{\mathsf{PostBQP}}
\newcommand{\BPP}{\mathsf{BPP}}

\newcommand{\NP}{\mathsf{NP}}

\newcommand{\QMA}{\mathsf{QMA}}

\newcommand{\PreciseQMA}{\mathsf{PreciseQMA}}
\newcommand{\PostQMA}{\mathsf{postQMA}}

\newcommand{\PreciseQCMA}{\mathsf{PreciseQCMA}}

\newcommand{\IP}{\mathsf{IP}}

\newcommand{\IPlimit}[2]{\mathsf{IP[#1,#2]}}
\newcommand{\MIP}{\mathsf{MIP}}
\newcommand{\NEXP}{\mathsf{NEXP}}
\newcommand{\PSPACE}{\mathsf{PSPACE}}
\newcommand{\BQPSPACE}{\mathsf{BQPSPACE}}

\newcommand{\PP}{\mathsf{PP}}
\newcommand{\sharpP}{\mathsf{\#P}}

\newcommand{\ket}[1]{| #1 \rangle}
\newcommand{\pair}[2]{\left(#1,#2\right)}

\newcommand{\abs}[1]{\left| #1 \right|}

\newcommand{\Tr}{\mathrm{Tr}}

\newcommand{\cL}{\mathcal{L}}

\newcommand{\bN}{\mathbb{N}}
\newcommand{\poly}{\mathrm{poly}}

\renewcommand{\Pr}{\mathrm{Pr}}

\newcommand{\qed}{$\square$}

\renewenvironment{proof}{%
  \noindent{\em Proof.\ }}{%
  \hspace*{\fill}\qed
  \vspace{2ex}}

\newtheorem{proposition}[theorem]{Proposition}
\newtheorem{claim}[theorem]{Claim}

\begin{document}
\setlength{\textheight}{8.0truein}    

\runninghead{ Towards a quantum-inspired proof for IP = PSPACE
}
            {A. Green, G. Kindler, and Y-P Liu }

\normalsize\textlineskip
\thispagestyle{empty}
\setcounter{page}{377}

\copyrightheading{2021}{5\&6}{2021} {0377--0386}

\vspace*{0.88truein}

\alphfootnote

\fpage{377}

\centerline{\bf
TOWARDS A QUANTUM-INSPIRED PROOF FOR $\IP=\PSPACE$}
\vspace*{0.37truein}
\centerline{\footnotesize
AYAL GREEN\footnote{ayalg@cs.huji.ac.il}, ~~~
GUY KINDLER\footnote{gkindler@cs.huji.ac.il}, ~~{\em and}~~~ 
YUPAN LIU\footnote{yupan.liu@gmail.com}}
\vspace*{0.015truein}
\centerline{\footnotesize\it 
School of Computer Science and Engineering}
\baselineskip=10pt
\centerline{\footnotesize\it The Hebrew University, Jerusalem, Israel}
\vspace*{0.21truein}
\publisher{September 16, 2020}{February 10, 2021}

\vspace*{0.21truein}

\abstracts{
We explore quantum-inspired interactive proof systems where the prover is limited. Namely, we improve on a result by \cite{AG17}  showing a quantum-inspired interactive protocol ($\IP$) for $\PreciseBQP$ where the prover is only assumed to be a $\PreciseBQP$ machine, and show that the result can be strengthened to show an $\IP$ for $\NP^{\PP}$ with a prover which is only assumed to be an $\NP^{\PP}$ machine - which was not known before. We also show how the protocol can be used to directly verify $\QMA$ computations, thus connecting the sum-check protocol by \cite{AAV13} with the result of \cite{AG17,LFKN90}. 
Our results shed light on a quantum-inspired proof for $\IP=\PSPACE$, as $\PreciseQMA$ captures the full $\PSPACE$ power. 
}{}{}

\vspace*{10pt}

\vspace*{3pt}
\communicate{R~Cleve~\&~M~Mosca }

\vspace*{1pt}\textlineskip    

\section{Introduction}
The study of interactive proof systems began in the 1980's, and while initially only a few non-trivial proof systems were known \cite{Bab85,GMR89}, at the beginning of the 1990's it was discovered that interactive proof systems are actually extremely powerful. 

In broad terms, in an interactive proof for a language $\cL$, a computationally weak \textit{verifier} interacts with one or more  \textit{provers} which are stronger computationally. For a given input $x$, the provers claim that $x\in \cL$, but the verifier would not just take their word for it. Instead, an interactive protocol is commenced, followed by the verifier either 'accepting' the claim, or 'rejecting' it. The protocol has perfect completeness if, when $x$ is indeed in $\cL$ and the provers honestly follow the protocol, the verifier eventually accepts. The protocol has soundness-parameter $s$, if when $x$ is not in $\cL$ then the verifier rejects with probability at least $1-s$, independently of whether the provers follow the protocol or not. If a protocol for $\cL$ exists that has perfect completeness and soundness $1/2$, we say that it is an interactive protocol for $\cL$.

The celebrated $\IP=\PSPACE$ \cite{LFKN90,Sha90} result showed that any language in $\PSPACE$ has an interactive protocol with a (randomized) polynomial-time Turing machine as a verifier, and with a single prover that is computationally unbounded. The result $\MIP=\NEXP$ \cite{BFL91} that soon followed, showed that if two provers are allowed instead of one, interactive protocols exist for any language in non-deterministic exponential time. For the languages in the smaller class $\NP$, the celebrated PCP theorem  \cite{ALMSS98,AS98} showed that they can be verified by a multi-prover interactive protocol with a single round: When proving that $x\in \cL$, the verifier sends each of two provers a 'question' string of $O(\log |x|)$ length, and gets a constant number of bits from each prover. It is important to note that in all multi-prover protocols, the provers may not communicate at all while the protocol is taking place.

\paragraph{Computation delegation.} The PCP theorem had a huge impact for showing hardness-of-approximation for optimization problems. But in addition, the PCP theorem as well as other interactive proofs (and the techniques used to obtain them) seem very relevant for another practical application, namely computation delegation. The goal in computation delegation, introduced in \cite{GKR08}, is for the verifier to commission the provers to perform a computational task. The provers must then supply a result of the computation, and also prove to the verifier, via an interactive protocol, that the result that they sent is indeed correct. Of course, this would be pointless if the proof protocol would require more resources than what the verifier needs to perform the computation herself. 

For delegation, the aforementioned line of results are not useful in their original form, as they usually assume provers that are computationally unbounded. For practical applications we would like even a relatively weak honest prover to be able to follow the protocol. 

Note that a similar question, namely of interactive \textit{arguments} is studied in cryptography (see e.g. \cite{CP13}). However in arguments, it is ok if the provers can convince the verifier of a false statement, if this requires the provers to be very strong computationally. Here we want the protocol to be secure even if the provers are very strong, but maintain that honest provers do not need to be as strong. 

\paragraph{Quantum delegation.} The motivation for delegation of quantum computing is quite clear. Suppose we are given a box that we can feed with inputs, and then obtain from it outputs that are claimed to be the result of a quantum computation. How can we tell if the box actually performs the desired computation?

The holy grail of quantum computation-delegation would therefore be a protocol that verifies a polynomial-time quantum computation, by letting a 
classical polynomial-time randomized machine interact with an efficient quantum machine, namely one which is only assumed to be have the power of $\BQP$. 

There has been significant progress on interactive quantum proofs and quantum computation delegation in the last decade. For instance, it was shown that a polynomial-time classical machine can verify a polynomial-time quantum computation, if she is allowed to interact with a constant number of entangled polynomial-time quantum provers \cite{RUV13,Ji16,NV17,CGJV19,NV18}\footnote{See \cite{JNVWY20} for erratum and correction regarding \cite{NV18}.}.
Likewise, if we consider a single very restricted quantum device, namely one with a constant number of qubits, then it can verify polynomially long quantum computations \cite{ABOE10,BFK09,FK17,ABOEM17,FHM18}.

Mahadev's celebrated result \cite{Mah18} (see also \cite{GV19,ACGH19,CCY19}) provides a protocol by which a strictly classical verifier can verify quantum computation using a single quantum prover.  
However Mahadev's protocol relies on a computational-hardness assumption for soundness (a standard post-quantum cryptographic assumption). A protocol that is unconditionally sound is still unknown. 

\paragraph{Towards unconditional quantum delegation. }
As an intermediate step in the search for information-theoretic soundness for verifying polynomial quantum computation with only classical interactions, a more general question arises: which classes of problems can be verified using a single prover which is only assumed to have the computational power of the same class? 

\begin{definition}[In-class IP - informal]
	\label{def:in-class-IP}
	Let $\IPlimit{\mathcal{P}}{\mathcal{V}}$ be the set of all languages $\cL$ for which there exists an efficient interactive proof protocol between a $\mathcal{V}$-power verifier and a prover, such that a prover in class $\mathcal{P}$ can follow the protocol for an input in $\cL$, and the protocol is sound against {\em any} prover when the input is not in $\cL$.
\end{definition}

A recent attempt by Aharonov and Green \cite{AG17} provides an interactive proof protocol verifying \textit{precise} polynomial-time quantum computations ($\PreciseBQP$), showing that $\PreciseBQP\subseteq\IPlimit{\mathcal{\PreciseBQP}}{\BPP}$.  Let us remind the reader of  the definition of $\PreciseBQP$:
Recall that in normal polynomial-time quantum computation ($\BQP$), the gap between the acceptance probability of $yes$-cases and $no$-cases is inverse-polynomial\footnote{This gap can be amplified to a constant using standard gap amplification techniques.}. In the precise case, however, this gap is inverse exponential. This change in the precision gives a much stronger class. In fact, it captures the full power of $\PP$ \footnote{Note that $\PP=\PreciseBQP=\PostBQP$. } \cite{Aar05,Kup15,GSSSY18}. We note that it immediately follows from  \cite{AG17}, that the same protocol can be extended to work for languages in $\p^\PreciseBQP=\p^{\PP}=\p^{\sharpP}$. Thus it gives a  quantum analogue of \cite{LFKN90}, which showed an interactive protocol for $\p^{\#\p}$ .

\paragraph{Our contribution.} To achieve quantum delegation, we eventually would like a protocol with a $\BQP$ prover, but it is also  natural to try to extend the protocol of \cite{AG17,LFKN90} to larger 
complexity classes. We do know that such a protocol exists for $\PSPACE=\BQPSPACE$ \cite{Sha90,Wat03}, but even getting a direct quantum-inspired protocol for $\BQPSPACE$ would be interesting. 
A potential way to achieve this goal is by considering $\PreciseQMA$ (informally, a variant of $\QMA$ where the gap between the acceptance probabilities of \textit{yes} and \textit{no} cases is inverse exponential), which has been shown to be equal to $\PSPACE$ \cite{FL16,FL18}. 

In this paper, we make a step towards a direct $\BQPSPACE\subseteq\IP$ result: we show that $\PreciseQCMA$, a sub-class of $\PreciseQMA=\BQPSPACE$, can be verified by a classical interaction with a $\PreciseQCMA$ prover.
Roughly speaking, $\PreciseQCMA$ is a quantum analogue of $\NP$ where the witness is still classical and the verifier is quantum, but there might be an exponentially small gap between the acceptance probability of \textit{yes} and \textit{no} instances.

\paragraph{Paper organization.} In Section 2 we prove our main result. We explain why the same technique and other attempts do not work for $\PreciseQMA$ in Section 3. 

\section{An in-class interactive proof for $\PreciseQCMA$}

Before stating our main result formally, let us define  $\PreciseQCMA$ \cite{GSSSY18}. 

\begin{definition}[$\PreciseQCMA$ \cite{GSSSY18}]
	\label{def:PreciseQCMA}
	A language $\cL=(\cL_{yes},\cL_{no})$ is in $\PreciseQCMA(c,s)$ for polynomial-time computable functions $c,s:\mathbb{N}\mapsto[0,1]$ if there exist polynomially bounded functions $k,p,m:\bN\mapsto\bN$ such that $\forall \ell\in \bN$, $c(\ell)-s(\ell) \geq 2^{-p(\ell)}$, and there exists a polynomial-time uniform family of quantum circuits $\{V_x\}_{|x|\in\bN}$ acting on $k(|x|)+m(|x|)$ qubits that takes a classical proof $y\in\binset^{m(|x|)}$ and outputs a single qubit. Moreover, the number of working qubits of $V_x$ is $k(|x|)$, and for any binary input $x$:
	\begin{itemize}
		\item{\bf Completeness:} If $x\in \cL_{yes}$, then $\exists y$ such that $V_x$ accepts $y$ with probability at least $c(|x|)$. 
		\item{\bf Soundness:} If $x\not\in \cL_{no}$, then $\forall y$, $V_x$ accepts $y$ with probability at most $s(|x|)$. 
	\end{itemize}

	Furthermore, $\PreciseQCMA := \cup_{c,s \geq \exp(-\poly(|x|))} \PreciseQCMA(c,s)$. 
\end{definition}

We can now state our main theorem formally. 

\begin{theorem}
	\label{thm:PreciseQCMA-in-IP}
	$\PreciseQCMA \subseteq \IPlimit{PreciseQCMA}{BPP}$.
\end{theorem}

Noting that $\PreciseQCMA$ is equivalent to the classical complexity class $\NP^{\PP}$\cite{MN17,GSSSY18}, we get the  following immediate corollary. 

\begin{corollary}
	\label{corr:NP^PP-in-IP}
	$\NP^{\PP} \subseteq \IPlimit{\NP^{\PP}}{BPP}$.
\end{corollary}

This result is an improvement on \cite{LFKN90,AG17}, and was not known before. 

\subsection{Proof of Theorem \ref{thm:PreciseQCMA-in-IP}}

\paragraph{An in-class interactive proof for $\PreciseBQP$. } 
Using the terminology of Definition \ref{def:in-class-IP}, the protocol in \cite{AG17} shows that 
	$$\PreciseBQP\subseteq\IPlimit{\PreciseBQP}{\BPP}.$$
Namely, they show a protocol by which a polynomial probabilistic classical machine can verify the acceptance probability of a polynomial quantum circuit\footnote{I.e. a quantum circuit consisting of a polynomial number of sequential local gates on $n$ qubits.}, to within inverse-exponential accuracy, by interacting with a $\PreciseBQP$ prover.

\paragraph{First step towards an in-class interactive proof for $\PreciseQCMA$. }
In order to verify whether $x$ is a \textit{yes} or \textit{no} instance of a language $\cL\in\PreciseQCMA$, one's first attempt would be to have the verifier ask the prover for a witness $w$ for $x$, and then use an in-class interactive proof for $\PreciseBQP$, such as the AG protocol, to verify the acceptance probability of $w$. 
Since $\cL\in\PreciseQCMA$ we indeed know that if $x\in \cL$ then there really is a classical witness that could be sent to the verifier, and that given the witness, the process of verifying it is a computation in the class $\PreciseBQP$, which can be verified via the AG protocol. 

If we assume that the prover can send the witness, this would complete the proof of Theorem \ref{thm:PreciseQCMA-in-IP}. 
However there is a caveat: while we allow the prover to be a $\PreciseQCMA$ machine, this only means that for a language $\cL\in\PreciseQCMA$ and an instance $x$ of $\cL$, the prover can distinguish whether $x$ is a \textit{yes} or a \textit{no} instance. But this does not a-priori mean that the prover can actually find the witness, even if it knows that $x\in \cL$. 

\paragraph{The issue with adaptive search. } 
For languages in $\NP$, we know how to utilize an $\NP$ oracle in a witness-finding algorithm: we use consecutive accesses to the oracle to adaptively find one bit of the witness at a time. For example, to find the first bit of the witness the verifier asks the prover if the following $\NP$ claim, denoted  $S_0$, is true: "there exists a valid witness for the instance $x$ where the first bit is $0$". If the answer is "no", the verifier can ask about the statement $S_1$, where the value of the bit is flipped. Once the first bit $b$ of the witness is found in this way, the verifier can continue to find the second bit by asking about the statements $S_{b0}$, etc. . 

The previous process indeed works for any instance of any language in $\NP\subseteq \PreciseQCMA$, but does it extend to all of $\PreciseQCMA$? Let $\cL$ be a language in $\PreciseQCMA$ and $x\in \cL$ be a \textit{yes} input. Consider e.g. a statement $S_{0}$ for $x$ as described above. Is this really a $\PreciseQCMA$ statement? This turns out to be somewhat subtle. The problem stems from the following situation: suppose $V_\cL$ is a $\PreciseQCMA$ verifier for $\cL$ with completeness $c$, $x$ is in $\cL$, but there are no 'valid' witnesses for $x$ that begin with $0$. In this case $S_{0}$ is false, but there might still be witnesses that begin with the $0$ bit and make $V_\cL$ accept with a probability that is arbitrarily close to $c$, and thus cannot be distinguished from an actually valid witnesses by a $\PreciseQCMA$ oracle. In other words - the adaptive queries to the oracle do not necessarily follow the $\PreciseQCMA$ promise. 

In Lemma \ref{lemma:finding-witness} below, we explain this problem in detail and specify a way for $\PreciseQCMA$ oracle to still find a correct classical witness for an instance $x$ that belongs to a $\PreciseQCMA$ language $\cL$. 

\begin{lemma}[adaptive search for a $\PreciseQCMA$ witness]
	\label{lemma:finding-witness}
	
	Given a language $\cL=(\cL_{yes},\cL_{no})\in\PreciseQCMA$, and access to a $\PreciseQCMA$ oracle, there is a classical efficient algorithm that finds a witness for any instance $x \in \cL_{yes}$. 
\end{lemma}

This completes the proof of Theorem \ref{thm:PreciseQCMA-in-IP}. 
\hfill
$\square$

\subsection{Proof of Lemma \ref{lemma:finding-witness}}

\begin{proof}
 The protocol simply consists of iteratively asking the oracle for the witness bit by bit as described above. To show that a $\PreciseQCMA$ oracle can indeed answer with the correct bits for such an adaptive search, we first formally define a language $\cL'$ associated with the problem of "finding the next bit of the witness" on the $i$-th round. 
Let $V_{\cL}$ be a $\PreciseQCMA$ verifier for $\cL$ with completeness $c$ and soundness $s$ (and assume $c-s\geq \exp(-\poly(|x|))$). Let $\cL':=\cL'_{V_{\cL}}$ be the language of pairs $\pair{x}{w_0}$ as follows: $\cL':=\{(x,w_0) | \exists w=w_0\circ w_1 \text{ s.t. }$ $\Pr[V_{\cL}\text{ accepts }(x,w_0\circ w_1)]$ $\geq c \}$. That is, $\cL'$ consists of \textit{yes} instances of $\cL$, paired with all the prefixes of correct witnesses associated with them. 

Given access to an oracle for $\cL'$ and a string $x$, it is easy to see that the previous adaptive search algorithm can be efficiently used to find a witness for $x$ if one exists. So in order to prove Lemma \ref{lemma:finding-witness}, it remains to show that $\cL'$ is in $\PreciseQCMA$, meaning the oracle associated with $\cL'$ is indeed a $\PreciseQCMA$ oracle\footnote{We note that in order to regard the language $\cL'$ as a promise problem, we implicitly considered  $\cL'_{no}$ as the complement of $\cL'$. That is, $\cL'_{no}:=\{(x,w_0) | \forall w=w_0\circ w_1$ : $\Pr[V_{\cL}\text{ accepts }(x,w_0\circ w_1)]$ $< c \}$.}. 

 At first glance, it seems straightforward to construct a $(c,s)$-$\PreciseQCMA$ verifier $V_{\cL'}$ where $c-s\geq \exp(-\poly(n))$ for $\cL'$: as we already have $V_{\cL}$ and given an instance $(x, w_0)$ and a witness $w_1$, $V_{\cL'}$ would just run $V_{\cL}$ on the instance $x$ with witness $w=w_0\circ w_1$. 
It is easy to see that $\forall x \in \cL_{no}$ and any $w_0$, $V_{\cL'}$ will accept $(x,w_0)$ with probability at most $s$, since no correct witness exists; also $\forall x \in \cL_{yes}$, $V_{\cL'}$ will accept $(x,w_0)$ for every $w_0$ such that $w_0$ is a prefix of a valid witness $w_0\circ w_1$ with probability at least $c$ if given $w_1$ as a witness. 

However, what about the case when $V_{\cL'}$ is given $(x,w_0)$ where $x \in \cL$ but $w_0$ is not a prefix of any correct witness? 
In that case, the verifier may seemingly accept with probability smaller than but \textit{arbitrarily} close to $c$ although $(x,w_0)$ is not in $\cL'$ - this would mean that $V_{\cL'}$ does not have a proper separation between its completeness and soundness, and so it is not sufficient for proving that $\cL' \in \PreciseQCMA$. 

Fortunately, we can show that there exists a verifier $V_{\cL}$ for $\cL$ such that the probability with which it accepts \textit{any} input-witness pair $(x,w)$ is either at least some parameter $c$, or bounded from above by $s$, where $c-s \geq \exp(-\poly(|x|))$. Such a verifier will solve the previous issue and give us a $\PreciseQCMA$ verifier for $\cL'$. So it is only left to show the existence of such a verifier. This is done in Proposition \ref{prop:good-verifier-exists} below using similar techniques to those of \cite{JKNN12}, and completes the proof. 

\begin{proposition}
	\label{prop:good-verifier-exists}
    Let $\cL$ be a $\PreciseQCMA$ language. Then there exists a $\PreciseQCMA$ verifier $V^*_{\cL}$  for $\cL$ with completeness $c$ and soundness $s$ such that:
    \begin{enumerate}[(1)]
        \item $|c-s| \geq 2^{-\poly(n)}$, where $n$ is the length of the input $x\in \cL$. 
        \item For any $x\in\cL$ and \textit{any} classical witness $w$, either $\Pr[V^*_{\cL}\text{ accepts } (x,w)] \geq c$\\ or $\Pr[V^*_{\cL} \text{ accepts } (x,w)] \leq s$. 
    \end{enumerate}
\end{proposition}

\begin{proof}
    We start with any $\PreciseQCMA$ verifier $V_{\cL}$ for $\cL$. Given an input $x$ of size $n$, $V_{\cL}$ uses a circuit $G_{\cL,n}$ which consists of $l=\poly(n)$ (local) gates, with soundness and completeness parameters $s'$ and $c'$ s.t $\varepsilon' := s'-c'\geq 2^{-\poly(n)}$. The idea is to simulate $G_{\cL,n}$ by another circuit $G^*_{\cL,n}$ whose size is polynomial in the size of $G_{\cL,n}$, but uses only Hadamard and Toffoli gates. We will show that $G^*_{\cL,n}$ inevitably has the desired properties, and then simply define $V^*_{\cL}$ as using the circuit $G^*_{\cL,n}$ for inputs of size $n$. 
    
    By \cite{Shi02,Aha03} and \cite{KSV02}, $G_{\cL,n}$ can be approximated by a circuit $G^*_{\cL,n}$ such that the acceptance probabilities of $G^*_{\cL,n}$ and $G_{\cL,n}$ differ by at most $\varepsilon=\varepsilon'/4=2^{-\poly(l)}$ for any input and which uses only $t=l\cdot \poly\log(1/\varepsilon)$ Hadamard and Toffoli gates. Note that the size of $G^*_{\cL}$ is $\poly(l)$. Hence for any $x\in\cL$, there exists a witness $w$ such that $V^*_{\cL}$ accepts $(x,w)$ with probability at least $c=c'-\varepsilon$ and $\forall x\notin\cL$ and $\forall w$, $\Pr[G^*_{\cL,n}\ket{x}\ket{w} = \ket{\text{Accept}}] \leq s'+\varepsilon =: s$, and $c-s\geq \varepsilon'/2$. 
    
    \paragraph{Acceptance probabilities are dyadic numbers.} Next we show that there is only a discrete set of values that $\Pr[G^*_{\cL,n} \ket{x}\ket{w}=\ket{\text{Accept}}]$ can obtain, and that these values are inverse-exponentially separated. To see this, we look at the state $\sum_{y\in\{0,1\}^N} a_y \ket{y}$ which is the result of operating $G^*_{\cL,n}$ on a given pair $(x,w)$, prior to its measurement\footnote{Where $N=\abs{x}+\abs{w}$ is the total number of qubits the circuit acts on. We note that we did not mention an additional ancilla register, as it can be thought of as being part of the witness register $w$.}. W.l.o.g.,
    \[\Pr[G^*_{\cL,n}\ket{x}\ket{w}=\ket{\text{Accept}}] = \sum_{y:y_1=0} |a_y|^2.\] 
    
    \begin{claim}
    \label{claim:acc-prob-form}
    Each number $|a_y|^2$ is of the form $|a_y|^2=\frac{\alpha}{2^t}$ where $\alpha$ is integer valued, and $t$ is the number of gates in $G^*_{\cL,n}$. 
    \end{claim}

    \begin{proof}
     To see this, consider the state $(G^*_{\cL,n})^{(i)}\ket{x}\ket{w} = \sum_{y\in\{0,1\}^{N}} a^{(i)}_y \ket{y}$, which is obtained after applying the first $i$ gates of $G^*_{\cL,n}$ to $\ket{x}\ket{w}$. Furthermore, for each $y \in \{0,1\}^{N}$, and every $k\in[N]$, we define $y^k$ to be the same as $y$, except for a bit-flip in the $k$'th coordinate. And last, let $h_i$ be the number of Hadamard gates within the first $i$ gates of $G^*_{\cL,n}$. Now, a simple induction shows that for all $y \in \{0,1\}^{N}$, $a_y^{(i)}$ is of the form $a_y^{(i)} = \frac{\alpha_y^{(i)}}{2^{h_i/2}}$ where $\alpha_y^{(i)}$ is integer valued. 
     It certainly holds when $i=0$, as $(G^*_{\cL,n})^{(0)} \ket{x}\ket{w} = \ket{x,w}$. Going from $i-1$ to $i$, note that if the $i$-th gate of $G^*_{\cL,n}$ is a Toffoli, it does not affect the coefficient of the state except for reordering, and so the induction hypothesis still holds. If the $i$-th gate is Hadamard on the $k$-th coordinate, then all the resulting coefficients $a^{(i)}_y$ are of the form  
    \[
        a^{(i)}_y = \frac{1}{\sqrt{2}} \left(a^{(i-1)}_{y^k} \pm a^{(i-1)}_{y}\right). 
    \]
    And using the induction hypothesis:
    \[
        a^{(i)}_y = \frac{1}{\sqrt{2}}
        \left(\frac{\alpha^{(i-1)}_{y^k}}{2^{h_{i-1}/2}}
        \pm \frac{\alpha^{(i-1)}_{y}}{2^{h_{i-1}/2}} \right) 
        = \frac{\alpha^{(i-1)}_{y^k} \pm \alpha^{(i-1)}_{y}}{2^{(h_{i-1}+1)/2}} = \frac{\alpha^{(i)}_{y}}{2^{h_{i}/2}} 
    \]
    Where in the last transition we used the fact that $h_{i-1} + 1 = h_i$ because the $i$'th gate is Hadamard, and define $\alpha^{(i
    )}_{y} = \alpha^{(i-1)}_{y^k} \pm \alpha^{(i-1)}_{y}$ which is an integer as a sum of integers.  This concludes the induction, from which Claim \ref{claim:acc-prob-form} follows trivially.

\end{proof}

    As a corollary, by simply summing the probabilities of every standard basis element which has $0$ for its first qubit, we obtain $\forall x,w$, $|x|=n$, $\Pr[G^*_{\cL,n} \ket{x}\ket{w} = \ket{\text{Accept}}]$ is $\alpha/2^t$, for $t$ as defined earlier in Proposition \ref{prop:good-verifier-exists} proof and an integer valued $\alpha$.  
    
    \paragraph{Acceptance probabilities are exponentially separated. } We have shown that all acceptance probabilities of $V^*_{\cL}$ for a given input size are of the form $\alpha/2^t$. Let us denote $c^* := \min\{\frac{\alpha}{2^t}|\frac{\alpha}{2^t} \geq c\}$, and let $s^* := c^* - \frac{1}{2^t}$. It is obvious from the above discussion that 
    \begin{enumerate}[(1)]
        \item $\forall x \in \cL$, $\exists w$ s.t. $\Pr[V^*_{\cL} \ket{x}\ket{w} = \ket{\text{Accept}}] \geq c^*$. 
        \item $\forall x \in \bar{\cL}$, $\forall w$, $\Pr[V^*_{\cL} \ket{x}\ket{w} = \ket{\text{Accept}}] \leq s^*$. 
        \item $\forall x$, $\forall w$, if $\Pr[V^*_{\cL} \ket{x}\ket{w} = \ket{\text{Accept}}] < c^*$, then $\Pr[V^*_{\cL} \ket{x}\ket{w} = \ket{\text{Accept}}] \leq s^*$. 
    \end{enumerate}
    
    To finish the proof of Proposition \ref{prop:good-verifier-exists}, it is therefore only left to note that $c^*$ and $s^*$ are separated by some inverse exponential function of $n$. 
    \end{proof}

Proposition \ref{prop:good-verifier-exists} completes the proof of Lemma \ref{lemma:finding-witness}, and with that we are finished.

\end{proof}

\section{Discussion and Open Problems}

In the previous sections, we showed an in-class interactive proof for $\PreciseQCMA$ (Theorem \ref{thm:PreciseQCMA-in-IP}). 
One could ask if this protocol may be extended to $\PreciseQMA$. However, while in the case of  $\PreciseQCMA$ the prover could send the witness to the verifier, a description of a \textit{quantum} witness of polynomial length is, in general, exponentially long.  We could, instead, consider replacing the classical channel by a quantum one. But unfortunately even then, it is not clear how to obtain the required exponential precision from such a protocol without repeating it exponentially many times.
\vspace{8pt}

Besides $\PreciseQCMA$, the other natural sub-class of $\PreciseQMA$ is $\QMA$, where \textit{yes} instances and \textit{no} instances are separated by a constant, rather than exponentially small, probability gap. 
Below we state an interactive protocol for $\QMA$ and explain why a naive approach fails in extending it to $\PreciseQMA$:

\begin{proposition}
	\label{prop:QMA-in-IP}
	$\QMA \subseteq \IPlimit{\PreciseBQP}{\BPP}$.
\end{proposition}

The main idea in showing a simple protocol which proves Proposition~\ref{prop:QMA-in-IP} is using the witness-preserving gap amplification \cite{MW05,NWZ09} on the verification circuits for languages in $\QMA$. 

More accurately, given a language $\cL=(\cL_{yes},\cL_{no}) \in \QMA$, we have a verification circuit $V_\cL$ such that for \textit{yes} instance $x\in\cL_{yes}$, there exists an $m$-qubit state $\ket{\psi}$ s.t.  $\Pr[V_{\cL}\ket{\psi} = \ket{\text{Accept}}] \geq \frac{2}{3}$; 
but for \textit{no} instance $x\in\cL_{no}$, $\forall \ket{\psi}$, $\Pr[V_{\cL}\ket{\psi} = \ket{\text{Accept}}]  \leq \frac{1}{3}$. 

Using the witness-preserving gap amplification for $\QMA$, we can construct a new verification circuit $V'_{\cL}$, such that for $x\in\cL_{yes}$, $\Pr[V'_{\cL}\ket{\psi} = \ket{\text{Accept}}] \geq 1-2^{-(m+2)}$; but for $x\in\cL_{no}$, $\Pr[V'_{\cL}\ket{\psi} = \ket{\text{Accept}}] \leq 2^{-(m+2)}$. 
This means that if we take $\rho$ to be the maximally mixed state $2^{-m} I_m$ then for $x\in\cL_{yes}$, the probability $p_{accept} = \Tr\left(V'_{\cL} \rho {V'_{\cL}}^{\dagger}\right)$ of $V'_{\cL}$ accepting the state $\rho$ is bounded from below by $2^{-m}\cdot(1-2^{-(m+2)}) > 2^{-(m+1)}$. On the other hand, for $x\in\cL_{no}$ the probability $p_{accept}$ is bounded from above by $2^{-(m+2)}$. 
Hence, it is enough to evaluate $p_{accept}$ to within exponential precision in order to distinguish \textit{yes} instances from \textit{no} instances, which can be done by the protocol in \cite{AG17} (up to a straightforward purification). 

\medskip We note that the protocol of Proposition \ref{prop:QMA-in-IP} does not work for $\PreciseQMA$: the witness-preserving gap amplification technique of \cite{MW05,NWZ09} simulates running the verification $V_{\cL}$ several times with independent witnesses, while actually using the same witness each time. This way, a constant completeness-soundness gap of a $\QMA$ language can be efficiently amplified.  
However, to separate acceptance probabilities which are different only by an inverse exponential, exponentially many repetitions will be required. This will make the size of the verification circuit $V'_{\cL}$ exponential. 

\paragraph{$\PostQMA$. } We further note that $\PostQMA=\PSPACE$ (see \cite{MN17} for definition) seemingly does not have the problems mentioned above for $\PreciseQMA$, as the gap between the \textit{yes} and \textit{no} case acceptance probabilities is constant. 
However, in $\PostQMA$, the acceptance probabilities are \textit{conditioned} on measuring part of the computed register (which does not include the output qubit)  in a specific state. 
We do not know of a witness-preserving amplification technique which would support this kind of conditioning. 
In this regard, it is worthwhile mentioning that \cite{DGF20} recently showed that\footnote{Assuming $\PostQMA \not\subseteq \PP$. See Lemma. 21 in \cite{DGF20}.} witness-preserving gap amplification for $\PostQMA$ can only be done up to a point without increasing the size of the witness - in contrast to existing gap amplification techniques, where the amplification depends only on the size of the circuit.

\nonumsection{Acknowledgements}
\noindent
We thank Dorit Aharonov for helpful discussions. 
GK was supported by ISF Grant No. 1692/13 and ISF Grant No. 2635/19. 
AG and YL were supported by ISF Grant No. 1721/17. AG was also partially supported by CYBER grant. 

\nonumsection{References}
\noindent
\newcommand{\etalchar}[1]{$^{#1}$}

\end{document}